\documentclass[10pt,twoside]{siamltex}
\usepackage{amsfonts,epsfig}
\usepackage{amssymb}

\usepackage{xy}
\input{xy}
\xyoption{all} \xyoption{rotate}

\newtheorem{remark}[theorem]{Remark}

\newtheorem{example}[theorem]{Example}

\begin{document}


\bibliographystyle{plain}
\title{Strong connectivity and its applications}

\author{
Peteris\ Daugulis\thanks{Institute of Life Sciences and
Technologies, Daugavpils University, Daugavpils, LV-5400, Latvia
(peteris.daugulis@du.lv). } }

\pagestyle{myheadings} \markboth{P.\ Daugulis}{Strong connectivity
and its applications}

\maketitle

\begin{abstract} Directed graphs are widely used in modelling of nonsymmetric relations in various sciences
and engineering disciplines. We discuss invariants of strongly
connected directed graphs - minimal number of vertices or edges
necessary to remove to make remaining graphs not strongly
connected. By analogy with undirected graphs these invariants are
called strong vertex/edge connectivities. We review first
properties of these invariants. Computational results for some
publicly available connectome graphs used in neuroscience are
described.

\end{abstract}


\section{Introduction} \label{intro-sec}

\subsection{The subject of study}

Directed graphs are widely used to model various objects and
processses having nonsymmetric features. Directed paths and
related notions of strong connectedness and strongly connected
components may play important roles in applications. In this paper
we discuss invariants of strongly connected graphs which describe
stability of strong connectedness with respect to vertex and edge
removal operations. Strongly connected graphs having high values
of these invariants may be considered to be stably strongly
connected with respect to removal of vertices or edges. Therefore
strong connectivities may be useful in applications.

As an example of possible application we compute these invariants
for connectome graphs which are studied in neuroscience.

\subsection{Background}

\subsection{Review of related graph theory}

In this subsection we list the necessary notions and notations
from graph theory, see also \cite{D}.

\subsubsection{Directed graphs}

We will use notions of reachability between vertices, strong
connectedness relation in the set of vertices, strongly connected
components (SCC), condensation graph.

The underlying undirected graph of a directed graph $\Gamma$ is
denoted by $\mathcal{U}(\Gamma)$: nonempty sets of directed edges
between any vertices $u$ and $v$ are substituted by undirected
edges $u-v$.

For a modern treatment of theory of directed graphs see \cite{Ba}.

\subsubsection{Undirected graphs}

Let $\Gamma=(V,E)$ be a noncomplete undirected graph. We will
denote vertex/edge connectivity of $\Gamma$ by
$\zeta_{0}(\Gamma)$/$\zeta_{1}(\Gamma)$ (tradional notations:
$\zeta_{0}=\kappa$ and $\zeta_{1}=\lambda$).

The directed graph obtained from an undirected graph $\Gamma$ by
substituting each undirected edge $u-v$ by two directed edges
$u\rightleftarrows v$ is denoted as $\mathcal{D}(\Gamma)$. Note
that $\mathcal{U}(\mathcal{D}(\Gamma))=\Gamma$, but
$\mathcal{D}(\mathcal{U}(\Gamma))\neq \Gamma$, in general.

Given a vertex or edge subset $A$, we denote by $\Gamma-A$ the
graph obtained by removing $A$ from $\Gamma$. Given a vertex
subset $U$ we denote by $\Gamma[U]$ the $\Gamma$-subgraph induced
by $U$. These notations also make sense for directed graphs.

\section{Main results}

\subsection{Strong connectivity}\

Definitions given below can be found in \cite{Ba}.

\subsubsection{Definitions}\

\begin{definition} Let $\Gamma=(V,E)$ with $|V|\ge 2$ be a strongly connected directed graph.
We call a subset of vertices $U\subseteq V$ a \sl weakening vertex
subset\rm\ provided $\Gamma-U$ is not strongly connected or
$\Gamma-U$ has one vertex. We call a subset of edges  $D\subseteq
E$ a \sl weakening edge subset\rm\ provided $\Gamma-D$ is not
strongly connected.

\end{definition}

\begin{definition}
Let $\Gamma=(V,E)$ with $|V|\ge 2$ be a strongly connected
directed graph. Minimal cardinality of weakening vertex/edge
subsets of $\Gamma$ is called \sl strong vertex/edge connectivity
(SVC/SEC)\rm\ of $\Gamma$ (denoted as
$\sigma_{0}(\Gamma)$/$\sigma_{1}(\Gamma)$).

\end{definition}

\subsubsection{Properties}\

\begin{proposition}\label{1} $\Gamma=(V,E)$ - a strongly connected directed graph, $|V|\ge
2$.

\begin{enumerate}


\item $\sigma_{0}(\Gamma)\le \zeta_{0}(\mathcal{U}(\Gamma))$.

\item If $\Gamma=\mathcal{D}(\Delta)$ for some undirected graph
$\Delta$, then $\sigma_{0}(\Gamma)=\zeta_{0}(\Delta)$.

\end{enumerate}

\end{proposition}

\begin{proof}

\begin{enumerate}

%
%
%

\item Removal of a minimal $\mathcal{U}(\Gamma)$-disconnecting
vertex subset from $\Gamma$ makes the remaining graph
disconnected.

\item There are zero or two arrows in both directions between any
two vertices of $\Gamma$.  For any $U\subseteq V$ $\Gamma-U$ is
strongly connected iff $\mathcal{U}(\Gamma)-U$ is connected
undirected graph, hence the statement is proved.

\end{enumerate}
\end{proof}

\begin{proposition}\label{2}
Let $a,b\in \mathbb{N}$, $a\le b$. There exists a directed graph
$\Gamma_{a,b}$ such that $\sigma(\Gamma_{a,b})=a$ and
$\zeta_{0}(\mathcal{U}(\Gamma_{a,b}))=b$.

\end{proposition}

\begin{proof}
For any $b\in \mathbb{N}$
$\zeta_{0}(K_{b+1})=\sigma(\mathcal{D}(K_{b+1}))=b$. Thus the case
$a=b$ is proved.

Let $a<b$. We consider two subcases: $a\le \frac{b}{2}$ and
$a>\frac{b}{2}$.

\paragraph{Case $a\le \frac{b}{2}$}\

$\Gamma_{a,b}=(V_{a,b},E_{a,b})$, where
$V_{a,b}=\{U,V,W_{a},W'_{b-a}\}$, $|U|=|V|=b+1$, $|W_{a}|=a$,
$|W'_{b-a}|=b-a$. $E_{a,b}$ contains 1) edges from every element
of $U$ to every element of $W_{a}$, 2) edges from every element of
$W_{a}$ to every element of $V$, 3) edges from every element of
$V$ to every element of $W'_{b-a}$, 4) edges from every element of
$W'_{b-a}$ to every element of  $U$.

In this case $a=\min(a,b-a)$, $W_{a}$ is a minimal weakening
vertex subset for $\Gamma_{a,b}$. $W_{a}\cup W'_{b-a}$ is a
minimal disconnecting vertex subset of
$\mathcal{U}(\Gamma_{a,b})$.

\vspace{5pt}

\paragraph{Case $a> \frac{b}{2}$}\

$\Gamma_{a,b}=(V_{a,b},E_{a,b})$, where
$V_{a,b}=\{U,V,W_{a},W'_{b-a}\}$, $|U|=|V|=b+1$, $|W_{a}|=a$,
$|W'_{b-a}|=b-a$.
 $E_{a,b}$ contains 1) edges in both directions between every element of $U$ and every
element of $W_{a}$, 2) edges in both directions between every
element of $W_{a}$ and every element of  $V$, 3) edges from every
element of $U$ to every element of $W'_{b-a}$, 4) edges from every
element of $W'_{b-a}$ and every element of $V$.

In this case $W_{a}$ is a minimal weakening vertex subset for
$\Gamma_{a,b}$.

In both cases $W_{a}\cup W'_{b-a}$ is a minimal disconnecting
vertex subset of $\mathcal{U}(\Gamma_{a,b})$ and
$\zeta_{0}(\Gamma_{a,b})=|W_{a}\cup W'_{b-a}|=b$.

\end{proof}

\begin{proposition}\label{3}\
 SVC of a directed graph on $n$ vertices can be computed in
$O(n^5)$ time.

\end{proposition}

\begin{proof}
Assume again $\Gamma=(V,E)$ is an edge-weighted graph with weight
of each edge equal to $1$. For any two vertices $u,v$ define
$MF(u,v)$ equal to the maximal flow from $u$ to $v$ in $\Gamma$.
$MF(u,v)$ is equal to the cardinality of a maximal system of
disjoint $(u,v)$-paths. Define
$\sigma(u,v):=\min(MF(u,v),MF(v,u)).$ The minimal number of
vertices necessary to remove to make $u$ and $v$ in different SCC
(denoted by $\sigma(u,v)$) is equal to $\min(MF(u,v),MF(v,u))$, it
can be computed in $O(n^3)$ time, see \cite{C}. By statement 1 of
\ref{1} we have that $\sigma(\Gamma)=\min\limits_{u\in V,v\in
V}\sigma(u,v)$. Since there are $\frac{n(n-1)}{2}$ vertex pairs,
it follows that $\sigma(\Gamma)$ can be computed in $O(n^3)\cdot
n^2=O(n^5)$ time.
\end{proof}

\begin{example} $\Gamma_{1,3}$ is shown in Fig.1.  $W_{1}=\{1\}$, $W'_{2}=\{2,3\}$. A
minimal weakening vertex subset of $\Gamma_{1,3}$ is $W_{1}$. A
minimal disconnecting vertex subset of $\mathcal{U}(\Gamma_{1,3})$
is $W_{1}\cup W'_{2}$.

\begin{center} $ \xymatrix@R=0.8pc@C=0.8pc{
&&4\ar[rrrrrd]&&&&&&&&&&8\ar[llllldd]\ar[lllllddd]&&\\
5\ar[rrrrrrr]&&&&&&&1\ar[rrrrru]\ar[rrrrrrr]\ar[ddrrrrrrr]\ar[dddrrrrr]&&&&&&&9\ar[llllllld]\ar[llllllldd]\\
&&&&&&&2\ar[llllluu]\ar[lllllllu]\ar[llllllld]\ar[llllldd]&&&&&&&\\
6\ar[uurrrrrrr]&&&&&&&3\ar[llllluuu]\ar[llllllluu]\ar[lllllll]\ar[llllld]&&&&&&&10\ar[lllllllu]\ar[lllllll]\\
&&7\ar[uuurrrrr]&&&&&&&&&&11\ar[llllluu]\ar[lllllu]&&\\
} $

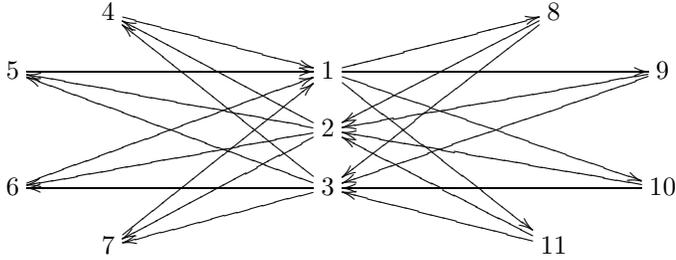
Fig.1.  - $\Gamma_{1,3}$.

\end{center}

\end{example}

\begin{example} $\Gamma_{2,3}$ is shown in Fig.2.  Pairs of oriented edges with common
vertices are shown as undirected edges. $W_{2}=\{1,2\}$,
$W'_{1}=\{3\}$. A minimal weakening vertex subset of
$\Gamma_{2,3}$ is $W_{2}$. A minimal disconnecting vertex subset
of $\mathcal{U}(\Gamma_{2,3})$ is $W_{2}\cup W'_{1}$.

\begin{center} $ \xymatrix@R=0.8pc@C=0.8pc{
&&4\ar@{-}[rrrrrd]\ar@{-}[rrrrrdd]\ar[dddrrrrr]&&&&&&&&&&8&&\\
5\ar@{-}[rrrrrrr]\ar@{-}[rrrrrrrd]\ar[ddrrrrrrr]&&&&&&&1\ar@{-}[rrrrru]\ar@{-}[rrrrrrr]\ar@{-}[rrrrrrrdd]\ar@{-}[rrrrrddd]&&&&&&&9\\
&&&&&&&2\ar@{-}[rrrrruu]\ar@{-}[rrrrrrru]\ar@{-}[rrrrrrrd]\ar@{-}[rrrrrdd]&&&&&&&\\
6\ar@{-}[rrrrrrruu]\ar@{-}[rrrrrrru]\ar[rrrrrrr]&&&&&&&3\ar[rrrrruuu]\ar[rrrrrrruu]\ar[rrrrrrr]\ar[rrrrrd]&&&&&&&10\\
&&7\ar@{-}[rrrrruuu]\ar@{-}[rrrrruu]\ar[urrrrr]&&&&&&&&&&11&&\\
} $

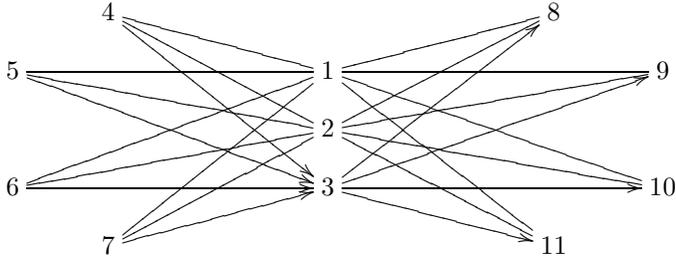
Fig.2.  - $\Gamma_{2,3}$, pairs of opposite edges shown as
undirected edges.

\end{center}

\end{example}

%
%
%
%
%
%
%
%
%
%
%
%
%
%
%
%
%
%

\begin{remark}
Computation of SC can be iterated to get more information about
graph structure as follows. Suppose we are given a strongly
connected graph $\Gamma$.

\begin{enumerate}

\item Find all (or at least one) minimal weakening set, find all
resulting $SCC$ $\Gamma_{1,1},\Gamma_{1,2},...,\Gamma_{1,n_{1}}$
and the corresponding condensation graph $\Delta_{1}$ with $n_{1}$
vertices.

\item Repeat Step 1 with each nontrivial SCC. Get a family of
acyclic condensation graphs and a family of SCC.

\end{enumerate}

\end{remark}

\subsection{An application}

It makes sense to study strong connectivity in graph models where
directed paths and reachability play important roles. A very
simple example of such an application would be a city graph model
where vertices are street intersections and edges are (possibly
one-way) streets. A meaningful question: how many one-way streets
or intersections can be closed so that the remaining street
network is still strongly connected.

In this subsection we describe some computational results related
to strong connectivity of connectome graphs considered in
neuroscience.

\subsubsection{Connectome graphs}

Connectome graphs are discrete mathematical models used for
modelling nervous systems on different scales, see \cite{K},
\cite{S2}. On the microscale level these graphs are special cases
of \sl cell graphs\rm\ which are studied for many types  of human
body tissues. In such graphs vertices correspond to cells and
edges correspond to physical cell contacts or substance transfer
(\sl volume transmission\rm , see \cite{Be}) links. See \cite{B}
for an example of cell graph application in tumour tissue
modelling. On mesoscale and macroscale levels connectome graphs
are essentially quotient graphs of microscale cell graphs. In this
paper we do not deal with connectome scale, nature of graph edges
and other modelling issues, we are interested only in applications
of graph-theoretic concepts and algorithms. Connectome graph edges
may be directed, undirected and weighted (labelled). We consider
only directed graph structure of connectomes, each edge is
assigned the constant weight $1$. We use data and references
available at  \cite{O}.

We assume that direction of connectome edges, directed paths and
reachability have considerable biological meaning. For example,
direction of edges may be correlated with flow of signals or
substances.

\subsubsection{Cat}


 Graph ($C$) description - macroscale connectome graph (cortiocortical connections, see \cite{R}),
strongly connected graph with $65$ vertices and $1139$ edges,
$\zeta_{i}(\mathcal{U}(C))=3$, diameter $3$, minimal degree $3$,
maximal degree $45$.

$\sigma_{0}(C)=\sigma_{1}(C)=1$. There is a unique weakening
vertex for $C$. After removing it $C$ splits into one trivial SCC
and a SCC having $63$ vertices. There is a unique weakening edge
pair for $C$. After removing it $C$ splits into two SCC having
$64$ and $1$ vertices.

We do $6$ iterations. At each step the graph splits into one
trivial and one large strongly connected component. The list of SC
of large components starting from $C$ is $[1,2,3,3,3,3,2]$. The
list of vertex connectivities of underlying graphs of large
components is $[3,3,7,7,7,6]$.

\subsubsection{Rat}

There are $3$ macroscale connectomes $\{R_{1},R_{2},R_{3}\}$ for
\sl Rattus nor\-ve\-gi\-cus\rm\ available at \cite{O}. Every graph
has one nontrivial SCC having $502$ or $493$ vertices. Graphs have
similar properties in all cases.
$\zeta_{i}(\mathcal{U}(R_{n}))=7$.

$\sigma_{0}(R_{n})=\sigma_{1}(R_{n})=2$. Each graph has a unique
weakening vertex pair. After removing a weakening vertex pair
$R_{n}$ splits into one or several trivial SCC and one large SCC.

%
%
%
%

\subsubsection{Fly}

A mesoscale fly connectome graph (see \cite{T}) description -
$1781$ vertices and $9735$ edges, $996$ strongly connected
components - one with $785$ vertices, one with $2$ vertices, the
other components trivial, minimal degree $1$, maximal degree
$927$.

Let $F_{1}$ - SCC of the described graph with $785$ vertices.
$\zeta_{i}(\mathcal{U}(F_{1}))=1$.
$\sigma_{0}(F_{1})=\sigma_{1}(F_{1})=1$, $173$ weakening vertices.
After removing a weakening vertex $F_{1}$ splits into several
trivial and $1$ large SCC. $F_{1}$ has $245$ weakening edges.

%
%

\subsection{Conclusion}
We discuss study of invariants directed graphs - strong
connectivities. Vertex connectivity can be computed in
polynomial-time using network algorithms.

We present some computation results involving graph models in
biology - connectome graphs of nerve tissues.

Some of our observations concerning connectome graphs:
\begin{enumerate}

\item[1)] after removing a minimal weakening vertex  subset the
considered connectome graphs split into one nontrivial and one or
several trivial SCC;

\item[2)] strong connectivities are close to undirected
connectivities;

\item[3)] often there is a unique weakening vertex or edge subset.

\end{enumerate}

%
%
%
%
%
%

\section*{Acknowledgement} Computations were performed using the
computational algebra system MAGMA, see Bosma et al. \cite{Bo}.

\end{document}